\documentclass[10pt, journal, twocolumn]{IEEEtran}

\usepackage{epsfig,balance,cite,times,amssymb,url,amsbsy,bm}
\usepackage{algorithm,algorithmic,amssymb,amsmath}
\usepackage{graphicx,subfigure,footmisc,paralist}
\newtheorem{lemma}{Lemma}[section]

\newtheorem{remark}[lemma]{Remark}

\usepackage[only,tensf,]{rawfonts}

\newcommand{\tx}[1]{\text{tx}(l)}
\newcommand{\rx}[1]{\text{rx}(l)}

\newcommand{\separator}{
  \begin{center}
    \rule{\columnwidth}{0.3mm}
  \end{center}
}

\newenvironment{separation}
{ \vspace{-0.3cm}
  \separator
  \vspace{-0.25cm}
}
{
  \vspace{-0.5cm}
  \separator
  \vspace{-0.15cm}
}

\newcommand{\bi}{\begin{itemize}}
\newcommand{\ei}{\end{itemize}}
\newcommand{\be}{\begin{enumerate}}
\newcommand{\ee}{\end{enumerate}}

\newcommand{\beq}{\begin{eqnarray*}}
\newcommand{\eeq}{\end{eqnarray*}}
\newcommand{\beqn}{\begin{eqnarray}}
\newcommand{\eeqn}{\end{eqnarray}}

\begin{document}

% paper title
\title{REFIM: A Practical Interference Management in Heterogeneous Wireless Access Networks}

\author{
\authorblockN{
%1.simple
%Kyuho Son, Soohwan Lee, Yung Yi and Song Chong
%2.full
Kyuho~Son,~\IEEEmembership{Member,~IEEE}, Soohwan~Lee,~\IEEEmembership{Student Member,~IEEE}, Yung~Yi,~\IEEEmembership{Member,~IEEE} and Song~Chong,~\IEEEmembership{Member,~IEEE}
}
\thanks{
Manuscript received June 8, 2010; revised December 3, 2010. This work was supported by IT R\&D program of MKE/KEIT [KI002137, Ultra Small Cell Based Autonomic Wireless Network]. Some part of this work was presented at WiOpt 2010, Avignon, France. This work was performed while the first author was with KAIST as a Ph.D. candidate.

K. Son is with the Department of Electrical Engineering, Viterbi School of Engineering, University of Southern California, Los Angeles, CA 90089 (e-mail: kyuho.son@usc.edu).

S. Lee, Y. Yi and S. Chong are with the Department of Electrical Engineering, Korea Advanced Institute of Science and Technology (KAIST), Daejeon 305-701, Korea (e-mail: shlee@lanada.kaist.ac.kr, \{yiyung, songchong\}@kaist.edu).

Digital Object Identifier 10.1109/JSAC.2011.xxxxxx.
}}

% make the title area
\maketitle

\begin{abstract}
Due to the increasing demand of capacity in wireless cellular networks, the small cells such as pico and femto cells are becoming more popular to enjoy a spatial reuse gain, and thus cells with different sizes are expected to coexist in a complex manner. In such a heterogeneous environment, the role of interference management (IM) becomes of more importance, but technical challenges also increase, since the number of cell-edge users, suffering from severe interference from the neighboring cells, will naturally grow. In order to overcome low performance and/or high complexity of existing static and other dynamic IM algorithms, we propose a novel low-complex and fully distributed IM scheme, called {\em REFIM} (REFerence based Interference Management), in the downlink of heterogeneous multi-cell networks. We first formulate a general optimization problem that turns out to require intractable computation complexity for global optimality. To have a practical solution with low computational and signaling overhead, which is crucial for low-cost small-cell solutions, e.g., femto cells, in REFIM, we decompose it into per-BS (base station) problems based on the notion of {\em reference user} and reduce feedback overhead over backhauls both temporally and spatially. We evaluate REFIM through extensive simulations under various configurations, including the scenarios from a real deployment of BSs. We show that, compared to the schemes without IM, REFIM can yield more than 40\% throughput improvement of cell-edge users while increasing the overall performance by 10$\sim$107\%. This is equal to about 95\% performance of the existing centralized IM algorithm (MC-IIWF) that is known to be near-optimal but hard to implement in practice due to prohibitive complexity. We also present that as long as interference is managed well, the spectrum sharing policy can outperform the best spectrum splitting policy where the number of subchannels is optimally divided between macro and femto cells.
\end{abstract}

%%% Local Variables:
%%% mode: latex
%%% TeX-master: "JSAC"
%%% End:

\begin{IEEEkeywords}
Interference management, heterogeneous wireless access networks, femto cells, reference user, power control, user scheduling, feedback reduction, distributed algorithm;
\end{IEEEkeywords}

\IEEEpeerreviewmaketitle

%===================================================================
\section{Introduction}
\label{section_introduction}
%===================================================================

% 1. Demanding data traffic & conventional approached
Many researchers from networking and financial sectors forecast that by 2014, the total mobile data traffic throughout the world will grow exponentially and reach about 3.6 exabytes per month, 39 times increase from 2009 \cite{refCisco,refDataeverywhere}. Pushed by this explosive demand mainly from bandwidth-hungry multimedia and Internet-related services in broadband wireless cellular networks, communication engineers seek to maximally exploit the spectral resources in all available dimensions.

% 2. Promising small cells approach
Small cells such as pico and femto cells, seem to be one of the most viable and economic solutions \cite{refChandrasekharFemtocell}. Of recent significant interest is the femto cell designed for usage in a home or an office and deployed by users, utilizing user's Internet connection, e.g., cable or DSL (Digital Subscriber Line) Internet services as a backhaul. Macro and pico cells are controlled by mobile network operators and use dedicated backhauls. Small cells are also considered as a way of incrementally increasing coverage and/or capacity inside the initial deployment of macro cells. In addition to the advantages of spectrum reuse efficiency, the network operators are also attracted by financial benefits because small cells can reduce both capital (e.g., hardware) and operating (e.g., electricity, site lease and backhaul) expenditures.

% 3. More challenging ICI mgt problem in heterogeneous networks
In heterogeneous multi-cell networks with a mixture of macro and small cells, interference is a major obstacle that can impair the potential gain of small cells and its pattern is highly diverse \cite{ref3GPP_interference}, e.g., interferences in macro-to-macro, femto-to-femto, and macro-to-femto. As the number of small cells increases, the number of users at cell edges suffering from low throughput due to severe interference also grows. In particular, femto base stations (BSs) are installed in an ad-hoc manner without being planned by users, not the network operators, which also increases the technical challenges of interference management (IM).

% 4. Static/Dynamic IM
To mitigate interference, a traditional frequency reuse or more enhanced schemes such as fractional frequency reuse (FFR) \cite{refGiuliano} and its variations \cite{refSonDynamic2, ref3GPP_SFR} can be utilized. All these schemes represent static IM algorithms for pure macro cell networks, where a specific reuse pattern is determined a priori by the network operator at offline. However, in reality, BSs are not uniformly deployed over the network. In particular, femto BSs are purely controlled by users, which implies that they may often be installed by users and even existing ones may be turned on/off dynamically. Under this situation, having the static reuse pattern is naturally inefficient. Recently, several dynamic IM algorithms have been proposed to address this problem \cite{refDasDynamic, refSonAdaptive, refGjendemsjBinary, refVenturinoCoordinated, refStolyar2, refEricssonFemto, refChandrasekharUplink}. They can significantly improve the performance over the static schemes, but many of them suffer from prohibitively high complexity and message passing among neighboring BSs.

% Related work I: wireless
In a single-carrier multi-cell network setup, several BS coordination schemes \cite{refDasDynamic, refSonAdaptive, refGjendemsjBinary} have been proposed under the assumption of binary power control, i.e., each BS transmits data with its given maximum power or zero. There are several recent works in a multi-carrier multi-cell network \cite{refVenturinoCoordinated, refStolyar2}. Venturino {\em et al.} \cite{refVenturinoCoordinated} proposed several algorithms which perform multiple iteration loops in a slot for user scheduling and power allocation. Although they can achieve near-optimal performance, all of them are centralized algorithms which are too complex to be implemented in practice. Stolyar {\em et al.} \cite{refStolyar2} proposed algorithms that adjust BS powers much more slowly than per-slot user scheduling. This time-scale separation does simplify the problem solution and reduce the complexity, but may lead to non-negligible performance loss. To tackle the cross-tier interference between macro and femto cells, there also have been several approaches by Sprint, Ericson \cite{refEricssonFemto} and Chandrasekhar {\em et al.} \cite{refChandrasekharUplink} that adjust the transmit powers of femto BSs based on the relative locations with respect to the macro BS. The power control also has been handled in ad-hoc networks. Chiang {\em et al.} \cite{refChiangPower} showed that in high SINR (signal to interference plus noise ratio) regime nonconvex power control optimization problems can be transformed into convex optimization problems through a geometric programming technique. In addition, there has been research on mitigation interference in a MIMO (multiple input multiple output) setting \cite{refZakhourDistributed}.

% Related work I: wired DSM
The part of dynamic IM algorithms, in particular, the power control component, was studied in the context of multi-tone DSL networks (see \cite{refTsiaflakisDistributed} and the references therein for a nice survey). Indeed, the wired multi-tone DSL model with crosstalk can be interpreted as a special case of the wireless multi-carrier cellular model with inter-cell interference when (i) only one user exists per cell and (ii) wireless channels are stationary and (iii) distributed operations among BSs are not crucial. In fact, the power control of our proposed solution is motivated by ASB (Autonomous Spectrum Balancing) \cite{refTsiaflakisDistributed,refCendrillonAutonomous} in the DSL network that uses the idea of reference line. However, in the IM over multi-cell wireless networks, much more challenging issues still remain for practical implementation, e.g., joint operation with multi-user scheduling in each cell, dynamic selection of reference users over time-varying channels, and small message passing among neighboring cells.

% Fundamental challenges of dynamic IM
The fundamental challenges in the dynamic IM are that {\em (i)} BS power control problem itself (even if scheduled users in each BS are fixed) is formulated by a highly nonconvex optimization\cite{refVenturinoCoordinated, refStolyar2,refChiangPower}, {\em (ii)} it is tightly coupled with multi-user scheduling, and {\em (iii)} heavy message passing is usually required to coordinate BS powers. In this paper, we aim at developing an IM scheme consisting of joint power allocation and user scheduling which is practical in terms of low complexity and small message passing, but yet the scheme achieves near-optimal performance. Low complexity and small message passing is particularly essential for low-cost solutions such as femto BSs because they are typically made of cheap devices for price competitiveness and connected to the low-speed residential cable or DSL Internet connections, not to the high-speed dedicated backhauls.
%yet its performance gain is high towards near-optimality.

% Spectrum splitting/sharing policies
Another important issue in heterogeneous networks is the way of sharing spectrum between macro and femto cells. If network operators adopt a {\em spectrum splitting policy} where macro and femto cells orthogonally use the resource for convenience of implementation, they can be free from the macro-to-femto interference. However, such a splitting policy needs to determine the ratio of optimal splitting that varies depending on the configuration of cells, e.g., density and position, resulting in spectrum inefficiency. Alternatively, a {\em spectrum sharing   policy} between macro and femto cells can be adopted to maximally reuse the resource. However, the macro-to-femto interference may harm the system performance unless appropriate IM algorithms are employed. It will be interesting, especially for the network operators, to answer which of these two polices is better under what circumstances.

\smallskip
Motivated by the above, we propose a novel IM algorithm, called {\em REFIM} (REFerence based Interference Management), whose core features are summarized as follows.

\begin{compactenum}[1)]
\item With the notion of {\em reference user}, each BS can approximate the interference impact of all other cells with a single user to which the BS generates the most significant interference. This abstraction substantially simplifies the problem, resulting in the power control algorithm with low computational and signaling overhead.
\item Due to the nonconvexity of the power control problem, different initial power settings may lead to different solutions. We empirically show that running the power control algorithm with the powers used at the previous slot as the initial powers can have effect of removing multiple loops without much performance degradation.
\item In the original power control based on the reference user, it requires to feedback per-user information in a cell to the neighboring BSs at every slot. We reduce such a heavy message passing overhead over backhauls both temporally and spatially.
\item REFIM has a nice feature of {\em incremental deployability} that partial deployment in some specific regions, probably starting from the regions that experience small capacity due to severe interference, sufficiently increases the capacity in those regions, but not affecting other regions. This is a desirable property for network operators who cannot afford to upgrade the IM module in all BSs.
\item We also demonstrate that, as long as an appropriate IM such as REFIM is adopted, the performance of the spectrum sharing policy is much better than the best performance of the spectrum splitting policy that the number of subchannels is optimally divided between macro and femto cells.
\end{compactenum}

\medskip
The remainder of this paper is organized as follows. In Section \ref{System_Model}, we formally describe our system model and general problem. In Section \ref{section_proposed_framework}, we propose a reference user based power allocation and user scheduling with feedback reduction ideas. In Section \ref{section_ComplexityAnalysis}, we analyze the computational and signaling complexity. In Section \ref{section_simulation}, extensive simulations under various configurations demonstrate the performance of the proposed algorithm compared to previous algorithms. Finally, we conclude the paper in Section \ref{section_conclusion}. 
%===================================================================
\section{System Model and Problem Definition}\label{System_Model}
%===================================================================
\subsection{Network and Traffic Model}
We consider a wireless cellular network consisting of multiple heterogeneous BSs, where $\mathcal{N}_{macro}$, $\mathcal{N}_{pico}$ and $\mathcal{N}_{femto}$ are the set of macro, pico and femto BSs, respectively. Denote by $\mathcal{K} \doteq \{1, \ldots , K\}$ and $\mathcal{N} = \mathcal{N}_{macro} \cup \mathcal{N}_{pico} \cup \mathcal{N}_{femto} \doteq \{ 1, \ldots , N\}$ the set of users and BSs, respectively. BSs and users are equipped with one transmit and one receive antenna, respectively. Each user is assumed to be connected to a single BS. Denote by $\mathcal{K}_n$ the (nonempty) set of users associated with the BS $n$, i.e., $\mathcal{K} = \mathcal{K}_1 \cup \cdots \cup \mathcal{K}_N $ and $\mathcal{K}_n \cap \mathcal{K}_m = {\o}$, for $n \neq m$. A full buffer traffic model with infinite data packets in the queue for each user at its associated BS is used to consider best-effort traffic. Macro and pico BSs can in general exchange information very fast with each other because there are connections between them via high-speed wired dedicate backhauls directly or through a base station controller. On the other hand, femto BSs can be assumed to exchange information slowly through user's Internet connections as backhauls.
%We assume that the system is saturated, i.e., an infinite amount of data exists for each user at its associated BS.

\subsection{Resource and Allocation Model}
We consider a system where a subchannel is a group of subcarriers as the basic unit of resource allocation. Assume that there are $S$ number of subchannels and all BSs can use all the subchannels for data transmission, i.e., universal frequency reuse. Denote by $\mathcal{S} \doteq \{1, \ldots , S\}$ the set of subchannels. We focus on the downlink transmissions in the time-slotted system. At each slot, each BS needs to determine (i) \textit{which user is scheduled on each subchannel} and (ii) \textit{how much power is allocated for each scheduled user on each subchannel}.

{\em User scheduling constraint}: In regard to (i), denote by $\bm I_s (t) \doteq [ I_s^{k,n} (t): k \in \mathcal{K}, n \in \mathcal{N} ]$ the user scheduling indicator vector, i.e., $I_s^{k,n} (t)=1$ when BS $n$ schedules its associated user $k$ on subchannel $s$ at slot $t$, and 0 otherwise. Furthermore, we denote the user scheduled by BS $n$ on subchannel $s$ at slot $t$ by $k(n,s,t)$. Reflecting that at most only one user can be selected in each subchannel for each BS, we should have:
\begin{equation}\label{eq:I_constraint}
\sum_{k \in \mathcal{K}_n} I_s^{k,n} (t) \leq 1, \quad \forall n \in \mathcal{N}, s \in \mathcal{S}.
\end{equation}

{\em Power constraint}: In regard to (ii), denote the transmit power of BS $n$ on subchannel $s$ at slot $t$ by $p_s^n (t)$. The vector containing transmit power of all BSs on subchannel $s$ is $\bm p_s (t) \doteq [p_s^1 (t), \cdots p_s^N (t)]^T$. In parallel, the vector containing transmit powers of all subchannels for BS $n$ is $\bm p^n (t) \doteq [p_1^n (t), \cdots p_S^n (t)]^T$. Each BS is assumed to have the total power budget and spectral mask constraints:
\begin{eqnarray}\label{eq:P_constraint}
\sum_{s \in \mathcal{S}} p_s^n (t) \!\!\!&\leq&\!\!\! P^{n,max}, \quad
\forall n \in
\mathcal{N}, \\
p_s^n (t) \!\!\!&\leq&\!\!\! P_s^{n,mask}, \quad \forall n \in
\mathcal{N}, s \in \mathcal{S}.
\end{eqnarray}
In practice, a typical transmit power of macro BSs is around 43dBm, which is 20$\sim$30dBm higher than that of small BSs. For notational simplicity, the time-slot index $(t)$ is dropped unless confusion arises.

\subsection{Link Model}
In this paper, we do not consider advanced multiuser detection or interference cancellation, and hence the interference from other BSs is treated as noise. We focus on the spectrum level coordination\footnote{More performance gain may be achieved by canceling inter-cell interference using signal level coordination, such as CoMP (Coordinated Multi Point Transmission and Reception) addressed in the LTE-Advanced, which is beyond the scope of this paper.}, i.e., finding multi-channel power allocation of each BS in order to improve system performance by mitigating the interference. For a given power vector $\bm p_s$, the received SINR for user $k$ from BS $n$ on subchannel $s$ can be written as:
\begin{equation}\label{eq:receivedSINR}
\gamma_s^{k,n}(\bm p_s) = \frac{ g_s^{k,n} p_s^n } {\sum_{m \neq n}
g_s^{k,m} p_s^m + \sigma_s^k},
\end{equation}
where $p_s^n$ and $g_s^{k,n}$ representing the nonnegative transmit power of BS $n$ on subchannel $s$ and the channel gain between BS $n$ and user $k$ on subchannel $s$ during a slot, respectively; $\sigma_s^k$ is the noise power. The channel gain is time-varying and takes into account the path loss, log-normal shadowing, fast fading, etc.

Following the Shannon's formula, the achievable data rate [in bps] for user $k$ on subchannel $s$ is given by:
\begin{equation}
r_s^{k,n} (\bm p_s) = \frac{B}{S} \log_2 \left( 1 + \frac{1}{\Gamma}
\gamma_s^{k,n} (\bm p_s) \right),
\end{equation}
where $B$ denotes the system bandwidth; $\Gamma$ denotes the SINR gap to capacity which is typically a function of the desired bit error ratio (BER), the coding gain and noise margin, e.g., $\Gamma = \frac{-\ln(5 \text{BER})}{1.5}$ in M-QAM (quadrature amplitude modulation)\cite{refGoldsmithVariable}. Note that $r_s^{k,n}(\bm p_s)$ is the potential data rate when the user $k$ is scheduled for service by BS $n$ on subchannel $s$ and its actual data rate becomes zero when another user is scheduled, i.e., $r_s^{k,n}(\bm p_s, \bm I_s) = I_s^{k,n} \cdot r_s^{k,n}(\bm p_s)$. We assume that $\Gamma=1$ and drop $B/S$ mainly for simplicity, but our results can be readily extended to other values of $\Gamma$ and $B/S.$

\subsection{General Problem Statement}
Our objective is to develop a slot-by-slot joint power allocation and user scheduling algorithm that determines $\big( \bm p(t) \big)_{t=0}^\infty$ and $\big( \bm I(t) \big)_{t=0}^\infty$,  where $\bm p(t) \doteq (\bm p_s(t), s \in \mathcal{S})$ and $\bm I(t) \doteq (\bm I_s(t), s \in \mathcal{S})$. The long-term achieved throughput vector $\mathbf{R} = \left(R_k: k \in \mathcal{K} \right)$, where $R_k = {\displaystyle \lim_{t \rightarrow \infty}}\frac{1}{t} \sum_{\tau=1}^t \sum_{s \in \mathcal{S}} r_s^{k,n}(\bm p_s(\tau), \bm I_s(\tau))$, is the solution of the following optimization problem:
\begin{eqnarray}
  (\textbf{Long-term P}): & \displaystyle{\max} & \quad \displaystyle{ \sum_{k \in K} U_k(R_k) }\\%
  & \textrm{subject to}   & \quad  \mathbf{R} \in \mathcal{R},
\end{eqnarray}
%where the set $\mathcal{R} \subset \mathbb{R}_+^{K}$ is the set of all achievable rate vectors over long-term, referred to as {\em throughput region}. We assume the standard condition of continuously differentiability and strictly increasing concavity of $U_k(\cdot)$.
where $U_k(\cdot)$ is a concave, strictly increasing, and continuously differentiable utility function for user~$k$; $\mathcal{R} \subset \mathbb{R}_+^{K}$ is the set of all achievable rate vectors over long-term, referred to as {\em throughput region}.
%It has been shown that $\mathcal{R}$ is a closed, bounded convex set. Note that this convex optimization problem is challenging to solve, because the constraint set $\cal R$ is not only available to the users or the BSs, but we also aim at devising an instantaneous, distributed algorithm by running which the system generates the optimal solution.
Note that this optimization problem is challenging to solve, since we aim at devising an instantaneous, distributed algorithm even though the constraint set $\mathcal{R}$ is neither available to the users nor the BSs.

With the help of the stochastic gradient-based technique in \cite{refStolyarOn} that selects the achievable rate vector maximizing the sum of weighted rates where the weights are marginal utilities at each slot, it suffices to solve the following slot-by-slot problem ({\bf P}) which produces the long-term rates that is the optimal solution of the ({\bf Long-term P}).
\begin{eqnarray}
  \hspace{-1.0cm}(\bm{P}): &\hspace{-0.3cm}\displaystyle \max_{\bm p, \bm I} \quad &\hspace{-0.3cm} h(\bm p,
  \bm I) = \sum_{k \in \mathcal{K}} w_k \sum_{s \in \mathcal{S}}
  r_s^{k,n}(\bm p_s, \bm I_s) \label{eq:P_objective}\\
  &\hspace{-0.3cm}\displaystyle \text{subject to} \quad &\hspace{-0.3cm} \sum_{k \in \mathcal{K}_n}
  I_s^{k,n} \leq 1, \quad \quad  \forall n \in
  \mathcal{N}, s \in \mathcal{S}, \label{eq:P_constraint1}\\
  &\quad &\hspace{-0.3cm} \sum_{s \in \mathcal{S}} p_s^n \leq
  P^{n,max}, \quad \forall n \in \mathcal{N}, \label{eq:P_constraint2}\\
  &\quad &\hspace{-0.3cm} p_s^n \leq P_s^{n,mask}, \; \quad \quad \forall n \in
  \mathcal{N}, s \in \mathcal{S},\label{eq:P_constraint3}%
\end{eqnarray}
where $w_k > 0$ is the derivative of its utility $w_k = \frac{dU(R_k)} {dR_k}|_{R_k=R_k(t)}$ that can be interpreted as the weight of user $k$ for the slot. For example, we can set $w_k$ as the inverse of its average throughput ${1}/{R_k(t)}$ to achieve proportional fairness among users \cite{refLiuOpportunistic}.
%Since the scheduling indicators are integer variables and the system objective is tightly coupled by the transmit powers of all BSs and nonlinear (neither convex nor concave),
Since the system objective is a nonconvex function of the transmit powers and is also tightly coupled with integer variables of the scheduling indicators, the problem $(\bm{P})$ is a mixed-integer nonlinear programming (MINLP). Unfortunately, it is known in \cite{refCendrillonOptimal} that even the simplified problem, in which user scheduling issue is eliminated (i.e., $|\mathcal{K}_n|=1$ for all $n \in \mathcal{N}$), is computationally intractable. To find a global optimal solution, we need to fully search the space of the feasible powers for all BSs with a small granularity along with the all possible combinations of user scheduling. Thus, even for a centralized algorithm, it may not be feasible in practical systems to solve (${\bm P}$) at each slot.

%===================================================================
\section{REFIM: Reference User Based Interference Management}\label{section_proposed_framework}
%===================================================================
\subsection{Joint Power Allocation and User Scheduling}

\smallskip
\noindent{\bf \em User scheduling for fixed power allocation}
\smallskip

We now present our proposed approach to solve the problem $(\bm{P})$. Note first that for any given feasible power allocation, the original problem can be decomposed into intra-cell user scheduling problems.

\smallskip
\begin{lemma}
For any fixed feasible power allocation $\bm p$, the original problem $(\bm{P})$ can be reduced to $N \times S$ independent subproblems for each BS $n$ and subchannel $s$ as follows:
\begin{eqnarray}
&\displaystyle \max_{\bm I_s} \quad & \sum_{k \in \mathcal{K}_n} w_k
\cdot I_s^{k,n} \cdot r_s^{k,n}(\bm p_s) \label{eq:subproblem_objective} \\
&\displaystyle \text{subject to} \quad & \sum_{k \in \mathcal{K}_n}
I_s^{k,n} \leq 1. \label{eq:subproblem_constraint}
\end{eqnarray}
\end{lemma}
\begin{proof}
For the given power allocation $\bm p$, we can rewrite $h(\bm p, \bm I)$ as follows:
\begin{eqnarray*}
\hspace{-0.5cm} h(\bm p, \bm I) &=& \sum_{n \in \mathcal{N}} \sum_{k \in
\mathcal{K}_n} w_k \sum_{s \in \mathcal{S}} I_s^{k,n} \cdot
r_s^{k,n}(\bm p_s) \\
&=& \sum_{n \in \mathcal{N}} \sum_{s \in \mathcal{S}} \Big[\sum_{k
\in \mathcal{K}_n} w_k \cdot I_s^{k,n} \cdot r_s^{k,n}(\bm p_s)\Big].
\end{eqnarray*}
As $w_k$ and $r_s^{k,n}(\bm p_s)$ are given parameters, we only have to investigate dependencies among $I_s^{k,n}$. Since the constraint (\ref{eq:P_constraint1}) for the given BS $n$ and subchannel $s$ does not affect the other BSs and subchannels at all, the original problem can be decomposed and is equivalent to individually solving the $N \times S$ subproblems in (\ref{eq:subproblem_objective}) and (\ref{eq:subproblem_constraint}) for each BS and subchannel.
%Since the constraints (\ref{eq:P_constraint1}) on $I_s^{k,n}$ do not play a role across different BSs and subchannels, the original problem is equivalent to independently solving the $N \times S$ subproblems in (\ref{eq:subproblem_objective}) and (\ref{eq:subproblem_constraint}) for each BS and subchannel. This completes the proof.
\end{proof}

\smallskip
Accordingly, an optimal user scheduling algorithm under the given power $\bm p$ is easily obtained by
\begin{equation}
I_s^{k,n} \!=\! \left\{
\begin{array}{rl} {}
\!1, &\hspace{-0.15cm}\displaystyle \textrm{if } k = k(n,s) = \arg \max_{k \in \mathcal{K}_n} w_k \cdot r_s^{k,n} (\bm p_s),\\
\!0, &\hspace{-0.15cm}\textrm{otherwise}.%
\end{array}\right.
%\quad \forall n \in \mathcal{N}, s \in \mathcal{S}
\label{eq:user_scheduling}
\end{equation}

\medskip
\noindent{\bf \em Power allocation for fixed user scheduling}
\smallskip

For any given user scheduling $\bm I$, the original problem reduces to the following power allocation problem:
\begin{equation*}
\begin{split}
\displaystyle \max_{\bm p} \;\; \sum_{n \in \mathcal{N}} &\sum_{s \in \mathcal{S}} w_{k(n,s)}
\! \log_2 \!\! \left( \!1\!+\! \frac{g_s^{k(n,s),n}p_s^n}{\sum_{m \neq n}g_s^{k(n,s),m}p_s^m \!+\! \sigma_s^{k(n,s)}} \! \right) \\
\displaystyle \text{subject to} \quad &\sum_{s \in \mathcal{S}} p_s^n \leq P^{n,max}, \quad \forall n \in \mathcal{N},\\
\displaystyle \quad  & p_s^n \leq P_s^{n,mask}, \;\; \quad \quad \forall n \in \mathcal{N}, s \in \mathcal{S}.
\end{split}
\end{equation*}
Solving the above problem requires the knowledge of all interference channel gains across cells and noise power, forcing it to operate in a centralized fashion.

%\bigskip
To overcome this complexity and develop a distributed scheme with small message passing, we introduce the concept of {\em reference user}. Let $\mathcal{N}(n)$ be the set of neighboring BSs\footnote{If there is any chance that users in BS $n$ will handed over to certain adjacent BSs, then we consider such a set of BSs as the neighboring BSs of BS $n$, denoted by $N(n)$. This set can be determined a priori by mobile network operators at the time of deployment and/or maintained in online based on the signal strengths between the BSs.} of BS $n$, and further denote by $A(n,s)$ the set of all scheduled users on subchannel $s$ in $\mathcal{N}(n)$, i.e., $A(n,s) = \left\{ k(m,s) \;|\; m \in \mathcal{N}(n) \right\}$. The reference user for BS $n$ on subchannel $s$ is defined as the user among $A(n,s)$ which has {\em the strongest channel gain} between the user and BS $n$ on subchannel $s$. We separately denote the index of BS to which the reference user belongs and the reference user by $\text{ref}_s^n$ and $k(\text{ref}_s^n,s)$, respectively. We will elaborate on the way of choosing the reference users shortly at the end of this subsection.

%{\bf The reference user is the one who has {\em the strongest channel gain} on subchannel $s$ between BS $n$ and the scheduled user $k(m,s)$ in neighboring cells $m \in \mathcal{N}(n)$.}
%where the reference user is one who receives the largest interference from each BS $n$ on each subchannel $s$, which we denote by $k(\text{ref}_s^n,s).$ We will elaborate on the way of choosing the references users shortly in subsection \ref{subsection_reference_user}.

Once the reference user is fixed, each BS tries to find its own power allocation taking into account just one reference user per subchannel instead of solving the above problem considering all ($N$ number of) cochannel users in the network. This approximation comes from the intuition that just considering {\em the worst-case user} may be a good approximation of the case when all the users are included. The problem $(\bm{P_n})$ to be solved by each BS $n$ can be written as follows.
\begin{eqnarray}
(\bm{P_n}): &\hspace{-0.3cm} \displaystyle \max_{\bm p_n} &
\sum_{s \in \mathcal{S}} w_s^{n} \log_2 \! \Bigg( 1 \!+\!
\frac{g_s^{n}p_s^n}{\displaystyle \sum_{m \neq n}g_s^{n,m}p_s^m + \sigma_s^{n}}\Bigg) \\
&& \!+\! \sum_{s \in \mathcal{S}} w_s^{\text{ref}_s^n} \log_2 \! \Bigg( 1 \!+\!
\frac{g_s^{\text{ref}_s^n} p_s^{\text{ref}_s^n} }{\displaystyle \sum_{m \neq {\text{ref}_s^n}} g_s^{{\text{ref}_s^n},m}p_s^m + \sigma_s^{\text{ref}_s^n}}\Bigg) \nonumber\\
&\hspace{-0.7cm}\displaystyle \text{subject to}  & \sum_{s \in \mathcal{S}} p_s^n \leq P^{n,max}, \\ \label{eq:P_n_constraint}
&\hspace{-0.7cm} & p_s^n \leq P_s^{n,mask}, \; \quad \quad \forall s \in \mathcal{S}. %
\end{eqnarray}
Note that we replace the index of scheduled users $k(m,s)$  with the corresponding index of BSs $m$ to keep our notations simple.\footnote{We use the index of BSs instead of the index of scheduled users as follows: $w_s^n \leftarrow w_{k(n,s)}$, $g_s^{n} \leftarrow g_s^{k(n,s),n}$, $g_s^{n,m} \leftarrow g_s^{k(n,s),m}$, $\gamma_s^{n} \leftarrow \gamma_s^{k(n,s),n}$, $\sigma_s^n \leftarrow \sigma_s^{k(n,s)}$, $w_s^{\text{ref}_s^n} \leftarrow w_{k({\text{ref}_s^n},s)}$, $g_s^{{\text{ref}_s^n}} \leftarrow g_s^{k({\text{ref}_s^n},s),{\text{ref}_s^n}}$, $g_s^{{\text{ref}_s^n},m} \leftarrow g_s^{k({\text{ref}_s^n},s),m}$, $\gamma_s^{{\text{ref}_s^n}} \leftarrow \gamma_s^{k({\text{ref}_s^n},s),{\text{ref}_s^n}}$ and $\sigma_s^{\text{ref}_s^n} \leftarrow \sigma_s^{k({\text{ref}_s^n},s)}$.}

For any given user scheduling and reference user selection, the corresponding optimal power allocation must satisfy Karush-Kuhn-Tucker (KKT) conditions \cite{refBookBoyd}. In particular, let $\lambda^n$ denote the nonnegative Lagrange multiplier associated with the total power budget constraint \eqref{eq:P_n_constraint}. Then, the optimal $\lambda^n$ and $p_s^n$ must satisfy the following equalities:
\begin{equation}
\begin{split}
\displaystyle p_s^n = &\left[ \frac{w_s^n}{\lambda^n \ln 2+ t_s^{n}} - \frac{\sum_{m \neq n} g_s^{n,m} p_s^m + \sigma_s^n }{g_s^{n}} \right]_{0}^{P_s^{n,mask}}\!\!,\\
&\textrm{where} \quad t_s^{n} = \frac{w_s^{\text{ref}_s^n} g_s^{{\text{ref}_s^n},n} \gamma_s^{{\text{ref}_s^n}}}{\sum_{l \in \mathcal{N}} g_s^{{\text{ref}_s^n},l} p_s^l + \sigma_s^{\text{ref}_s^n}}\;,\label{eq:KKT_condition1}
\end{split}
\end{equation}
\begin{eqnarray}
\hspace{-3.5cm}\lambda^n \big( \sum_{s} p_s^n - P^{n,max} \big) = 0,\label{eq:KKT_condition2}
\end{eqnarray}
where $\left[\cdot \right]_a^b \doteq \min \left[ \max \left[\cdot, a \right], b \right]$ and $\gamma_s^{{\text{ref}_s^n}}(\cdot)$ is the received SINR as defined earlier in \eqref{eq:receivedSINR}. Note that the $t_s^n$ term can be interpreted as a taxation (or penalty) term. If the reference user is close to the BS $n$ (i.e., high interference channel gain $g_s^{{\text{ref}_s^n},n}$), then the value of $t_s^n$ increases. Consequently, BS $n$ lowers its power level to reduce the harm to the reference user.
%The $\lambda^n$ is a non-negative Lagrange multiplier associated with the total power budget constrain (\ref{eq:P_constraint2}) and has to be chosen such that the following complementary slackness condition is satisfied:

%Clearly, the modified problem $(\bm{P_n})$ is {\bf a still nonconvex \cite{refVenturinoCoordinated, refChiangPower}}. Therefore, (\ref{eq:KKT_condition1}) and (\ref{eq:KKT_condition2}) are the first-order \textit{necessary} conditions and there exist a corresponding duality gap between an optimal primal solution.

Since the modified problem $(\bm{P_n})$ is still nonconvex \cite{refVenturinoCoordinated, refChiangPower}, \eqref{eq:KKT_condition1} and \eqref{eq:KKT_condition2} are the first-order \textit{necessary} conditions. Therefore, there might exist a duality gap to the optimal primal solution. However, encouraged by the state-of-the-art asymptotic result \cite{refCendrillonOptimal} that the duality gap becomes zero when the number of subchannels is large, we develop an effective approximation algorithm for the problem $(\bm{P_n})$ based on the conditions (\ref{eq:KKT_condition1})-(\ref{eq:KKT_condition2}). Note that a fixed point equation of $p_s^n$ in (\ref{eq:KKT_condition1}) is a monotonic function of $\lambda^n$. Thus, it can be solved via a fast bisection method. Starting from an initial power allocation and $\lambda^n$, we calculate the power $p_s^n$ and taxation term $t_s^{n}$ in (\ref{eq:KKT_condition1}) for all subchannels. If the sum of updated power exceeds $P^{n,max}$, then $\lambda^n$ is increased. Otherwise, $\lambda^n$ is decreased. With the updated power, we repeat this until the equation (\ref{eq:KKT_condition2}) holds. If no positive value of $\lambda^n$ matches the equality, then $\lambda^n$ is set to be zero. In the latter case (interfering too much), the BS $n$ does not use all of its available power.

\smallskip
\begin{remark}\label{remark:taxation}
If the taxation term $t_s^{n}$ is ignored, our power allocation algorithm is reduced to the water-filling (WF) algorithm \cite{refPalomarPractical}, where each BS acts selfishly in order to maximize its own performance. By adding this term $t_s^{n}>0$, each BS operates in a social way by considering the reference user and lowers the water-filling level, which could lead to a globally better solution.
\end{remark}

\medskip
\noindent{\bf \em General algorithm description: joint user scheduling and power allocation}
\smallskip

TABLE \ref{tbl:extended_algorithm} describes a conceptual pseudo-code of the general algorithm for our problem. At each slot, each BS starts from a proper power allocation. For the given power, each BS first executes the user scheduling and then abstracts what is happening in the neighboring environment. For example, the concept of reference user can be used as one of the abstraction methods. For this given user scheduling and neighboring environment abstraction, the BS iteratively updates its power and effects to the neighboring environment until $\bm p$ converges or the maximum number of iterations is reached. Then each BS repeats the user scheduling and the neighboring environment abstraction for the updated power and goes into the power allocation loop again. This procedure is repeated until the user scheduling $\bm I$ converges or the maximum number of iterations is reached.

\begin{table}[]
\begin{center} \caption{General Algorithm Description}
\separation
\label{tbl:extended_algorithm}
\begin{enumerate}
\item[1:] Power initialization
\item[2:] {\bf repeat} (user scheduling loop):
\item[3:] \hspace{1cm}User scheduling
\item[4:] \hspace{1cm}Neighboring environment abstraction
\item[5:] \hspace{1cm}{\bf repeat} (power allocation loop):
\item[6:] \hspace{2cm}Update the effect to the neighboring environment
\item[7:] \hspace{2cm}Power allocation
\item[8:] \hspace{1cm}{\bf until} $\bm p$ converges or max \# of iterations is reached
\item[9:] {\bf until} $\bm I$ converges or max \# of iterations is reached
\end{enumerate}
\vspace{-0.1cm}\separation
\end{center}\vspace{-0.1cm}
\end{table}

%For the given power, the BS executes the user scheduling by (\ref{eq:user_scheduling}) and selects the reference users by (\ref{eq:reference_user}) up to $M$ users. For this given user scheduling and reference user selection, each BS $n$ iteratively updates its own taxation term $t_s^n$ and power $\bm p^n$ until $\bm p$ converges or the maximum number of iterations is reached. And then each BS repeats the user scheduling and the reference user selection for the updated power and goes into power allocation loop again. This procedure is repeated until the user scheduling $\bm I$ converges or the maximum number of iterations is reached.

The general algorithm not only has a prohibitively high computational complexity due to inner and outer loops, but also requires multiple information exchanges per slot between BSs to reflect the updated interference level followed by the updated power. However, the multiple feedbacks in a single slot are practically impossible because each user can send its own measurement information to the BS once per slot. To overcome this complexity, we will propose a simplified algorithm in subsection \ref{subsection_simplified}, which executes user scheduling and power allocation step-by-step without loops.

\subsection{Online Reference User Selection Method}\label{subsection_reference_user}
Based on the notion of reference user, each BS needs to consider the only one user instead of all scheduled users in the network on each subchannel.

From BS $n$'s point of view, although its transmit power interferes with all the scheduled users on the corresponding subchannel in other BSs, the effect will be critical especially for the the user who has the strongest channel gain (or the most influenced user from BS $n$). This dominant victim user can be found in the neighboring (or adjacent) BSs $\mathcal{N}(n)$. Therefore, we propose an online reference user selection method, in which each BS $n$ chooses the reference user on each subchannel $s$ as follows:

\begin{separation}
\vspace{0.1cm}
{\bf Reference User Selection Rule}
\vspace{-0.4cm}
\separator
\vspace{-0.4cm}
\begin{equation}
\begin{split}
&k(\text{ref}_s^n,s) \textrm{ is the reference user on subchannel $s$}, \\
&\displaystyle \textrm{ where } \text{ref}_s^n = \arg \max_{m \in \mathcal{N}(n)} g_s^{k(m,s),n}. \label{eq:reference_user}
\end{split}
\end{equation}
\vspace{-0.3cm}
\end{separation}
It is worthwhile mentioning that the reference user is independently and locally selected by each BS on each subchannel and thus no centralized coordination is necessary. Fig. \ref{fig:reference user selection} depicts an example of our online reference user selection procedure.

\begin{figure}[t]
\centering \epsfig{file=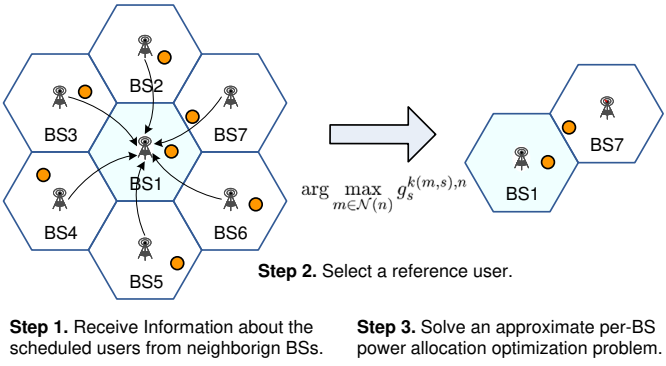,width=0.48\textwidth}
\caption{Online reference user selection method.}
\label{fig:reference user selection}
\end{figure}

The rule in \eqref{eq:reference_user} provides a guideline for choosing a reference user. There may be other methods such as (i) making one {\em virtual} user by averaging channels of the scheduled users from neighboring BSs, and (ii) selecting {\em multiple} reference users (e.g., select the $M$ worst users). As we will show later in subsection \ref{subsection_simulation1}, such variations do not lead to the high performance improvement, compared to the high increase in complexity.

\subsection{Feedback Reduction}
To determine a reference user at each slot, each BS $n$ requires {\bf (F0)} the channel gain $g_s^{k(m,s),n}$ of the scheduled users in neighboring cells $m \in \mathcal{N}(n)$. We call these users the {\em candidates} of reference users. Once the reference user is selected, to calculate a taxation term for power control, additional information about the reference user is necessary. The followings are the required information for the reference user:
\smallskip
\begin{compactenum}[\bf (F1)]
\item {The weight of the reference user $w_{k(\text{ref}_s^n,s)}$,}
\item {The received signal strength of the reference user $g_s^{k({\text{ref}_s^n},s),{\text{ref}_s^n}} p_s^{\text{ref}_s^n}$,}
\item {The noise plus interference strength of the reference user $\sum_{m \neq {\text{ref}_s^n}} g_s^{k({\text{ref}_s^n},s),m}p_s^m + \sigma_s^{k({\text{ref}_s^n},s)}$.}
\end{compactenum}
\smallskip

The information for the candidates needs to be collected by neighboring BSs and be forwarded to BS $n$. However, femto BSs use their Internet connection as backhauls, and thus message exchange may not be reliable and fast enough. Besides, there is no guarantee on the feedback latency. Even for macro BSs with a dedicated backhaul network, the per-slot message exchange may be a large overhead.
%In addition, the information for the candidate users except the elected reference user becomes unnecessary.
We present a more practical solution to reduce the backhaul feedback overhead both temporally and spatially.

%Instead of the per-slot message exchange for all the information, each BS calculates the time-average of the above information for all candidate users and sends them to its neighboring BSs {\em infrequently}, that is, every $T\gg1$ slots.

\smallskip
{\em Temporal feedback reduction:} Instead of the per-slot message exchange for all the information, (i) each candidate user first calculates the time-average of the information and send them to its associated BS infrequently (say, every $T\gg1$ slots), (ii) and then the BS broadcasts the information about all candidate users to its neighboring BSs through wired backhauls. The only thing that the BSs exchange at each slot is the indexes of the scheduled users. Each BS will maintain a table that contains these averages of candidate users. Once each BS receives the indexes of scheduled users from neighboring BSs, then it uses the information in the table corresponding to the indexes. Note that the slow feedback mechanism can be applied to femto BSs.

{\em Spatial feedback reduction:} First, we reduce the amount of infrequent feedback by making macro BSs send the information only for edge users. This idea comes from the intuition that the users in the center of cells are not likely to be selected as the reference user because the center users do not receive too much interference.

Second, we eliminate the per-slot feedback of indexes of scheduled users for femto BSs. The indexes of scheduled users (very small amount of information) can be easily exchanged between neighboring macro BSs at each slot through dedicated backhauls. However, this per-slot message exchange is not possible for femto BSs because their backhauls do not provide any guarantee on the feedback latency. To overcome this difficulty, we propose an alternative solution for the femto BSs that does not require per-slot message exchange at all. In the proposed solution, the femto BSs obtain the indexes of scheduled users by overhearing downlink control message (e.g, DL-map in the IEEE 802.16e \cite{refWiMAX}) from neighboring macro BS. This may need slight modification frame structure for the femto BS in the current wireless standards.

The remaining challenge is on the reverse direction, i.e., sending the indexes of scheduled users in the femto BSs to the neighboring macro BSs. We pay attention to one of the main features of femto cells, that is, the small coverage. In other words, the distances between users in a femto BS are relatively much shorter than the distances from the neighboring macro BSs. Thus, from neighboring macro BSs' perspective, it seems reasonable to assume that the users in the femto cell are spatially located at the same point. Based on this spatial simplification, the neighboring macro BSs simply can pick any user in the femto cell and consider as the candidate user.

\begin{table*}[t]
\begin{center} \caption{Complexity comparison of various algorithms}
{\renewcommand{\arraystretch}{1.3}
\label{tbl:computation_signaling_complexity}
\begin{tabular}{c|c|c|c|c} \hline
Algorithms & \multicolumn{2}{c|}{Computational complexity} & \multicolumn{2}{c}{Signaling complexity (inter-BSs)}\\ \hline%
& User scheduling & Power allocation & Per-slot feedback & Periodic feedback \\ \hline\hline%
EQ & $\mathcal{O}(SK)$ & {\em {\em Zero}} & {\em Zero} & {\em Zero}\\ \hline%
WF & $\mathcal{O}(SK)$ & $\mathcal{O}(S)$ or $\mathcal{O}\left(\log_2 \frac{P_{max}}{\epsilon} \right)$ & {\em Zero} & {\em Zero}\\ \hline%
MGR & $\mathcal{O}(SK)$ & $\frac{1}{n_p} \mathcal{O}(S) + n_v \mathcal{O}(SK) $ & {\em Zero} & $|N_{n}| S$ \\ \hline%
REFIM & $\mathcal{O}(SK)$ & $\mathcal{O}(S)$ or $\mathcal{O}\left(\log_2 \frac{P_{max}}{\epsilon} \right)$ & macro: $\rho S$ \& femto: {\em Zero} & $\rho |\mathcal{K}_n| A S$ \\ \hline%
MC-IIWF & \multicolumn{2}{c|}{$T_1\cdot\left(\mathcal{O}(SK)+\mathcal{O}(SN)\right)$} &
\multicolumn{2}{c}{Complete information is assumed.} \\ \hline
\end{tabular}}
\end{center}
\end{table*}

\subsection{Initial Power Setting}
Our algorithm requires an initial power value to compute the power allocation (see line 1 in Table I). Since our problem is a nonconvex problem, different starting points may lead to different solutions with different speeds. The following three strategies for the choice of initial power are carefully investigated in this paper.
\smallskip
\begin{compactenum}
\item {\em Uniform rule.} The power allocation always starts from the
  same point for every slot. Each BS uniformly splits its maximum
  transmission power to all subchannels, i.e., $p_{s}^{n,init} (t) =
  P^{n,max}/S$.
\item {\em Random rule.} Each BS randomly chooses the initial
  power level for each subchannel between 0 and $P^{n,max}/S$, and then
  each BS scales it up with an appropriate weight
  ${P^{n,max}}/{\sum_{s} p_s^{n,init}(t)}$ to use up the total
  transmission power budget, i.e., $\sum_{s} p_s^{n,init}(t) =
  P^{n,max}$.
\item {\em Previous rule.} Each BS starts from the power used at the
  previous slot, i.e., $p_{s}^{n,init} (t) = p_{s}^{n} (t-1)$.
\end{compactenum}
\smallskip
We will demonstrate later in subsection \ref{subsection_simulation1} that although we avoids multiple loops in a slot for power allocation and user scheduling, the performance of the previous rule is likely to remain unchanged while the other rules lose the performance much. This is because in some sense the previous rule exploiting temporal correlation has the effect of iterations for power allocation in a slot-by-slot manner. This result encourages us to design a simplified algorithm in subsection \ref{subsection_simplified}, which gets rid of multiple loops in a slot and executes user scheduling and power allocation sequentially with the previous rule.

\subsection{REFIM: Reference Based Interference Management} \label{subsection_simplified}
We now propose our final algorithm, called REFIM (REFerence based Interference Management), in Table \ref{tbl:proposed_algorithm} that merges all the components developed in the above. REFIM adopts the previous rule for initial power setting (see line 1), and uses the notion of reference user for the neighboring environment abstraction and limits the number of reference users to one (see line 3). While the general algorithm has user scheduling and power allocation loops (see lines 2 and 5 in Table \ref{tbl:extended_algorithm}), REFIM executes user scheduling (see line 2) and power allocation (see line 5) sequentially without loops. This step-by-step approach can not only be done very fast in a slot, but it requires the feedback from each user just once per slot. Through extensive simulations, such a simple algorithm will be shown to be efficient.

\begin{table}[h!]
\begin{center} \caption{REFIM: Reference based Interference Management}
\separation
\label{tbl:proposed_algorithm}
\begin{enumerate}
\item[1:] Power initialization ${\bm p^n(t)}\leftarrow{\bm p^n(t-1)}$
\item[2:] User scheduling according to (\ref{eq:user_scheduling}):\\
\hspace{0.3cm}$ \displaystyle k(n,s) = \arg \max_{k \in \mathcal{K}_n} w_k \cdot r_s^{k,n} (p_s^n)$.
\item[3:] Reference user selection according to (\ref{eq:reference_user}).
\item[4:] Taxation update according to (\ref{eq:KKT_condition1}):\\
%\smallskip \hspace{0.3cm}$\displaystyle t_s^{n} = \frac{w_s^{ref} g_s^{ref,n} \gamma_s^{ref} (\bm p_s)}{\sum_{l} g_s^{ref,l} p_s^l + \sigma_s^{ref}}.$ \smallskip
%
\item[5:] Power allocation via bisection:\\
$\left[a,b\right]\leftarrow[0,\lambda^n_{max}]$.\\
{\bf while} $|\sum_{s} p_s^n - P^{n,max}| < \delta$,\\
\hspace{0.3cm}Set $\lambda^n=(a+b)/2$ and update ${\bm p^n}$ according to (\ref{eq:KKT_condition1}).\\
\hspace{0.3cm}$\textrm{if } \sum_{s} p_s^n > P^{n,max},\;\; \textrm{then } \left[a,b\right]\leftarrow[\lambda^n,b]$,\\
\hspace{0.3cm}$\textrm{else } \sum_{s} p_s^n < P^{n,max},\;\; \textrm{then } \left[a,b\right]\leftarrow[a,\lambda^n]$.\\
{\bf end while}
\end{enumerate}
\vspace{-0.1cm}\separation
\end{center}\vspace{-0.1cm}
\end{table}

\section{Complexity Analysis}\label{section_ComplexityAnalysis}

In this section, we analyze the computational complexity and inter-BSs signaling complexity of REFIM compared to conventional equal power allocation (EQ) and selfish water-filling (WF), as well as MGR (Multi-sector GRadient) \cite{refStolyar2} and MC-IIWF (MultiCell Improved Iterative Water Filling) \cite{refVenturinoCoordinated}. Table \ref{tbl:computation_signaling_complexity} summarizes the results of complexity analysis.

Computational complexity consists of two parts: the complexity from user scheduling and power allocation. User scheduling has a linear complexity $\mathcal{O}(SK)$ with the number of users for each subchannel for all algorithms except MC-IIWF. For power allocation, EQ has zero complexity. WF can be obtained by either an exact algorithm $\mathcal{O}(S)$ or an iterative algorithm (i.e., a bisection method that converges to a solution with a certain error tolerance $\epsilon$) $\mathcal{O}\left(\log_2 \frac{P_{max}}{\epsilon} \right)$ \cite{refPalomarPractical}. The only difference between WF and REFIM is the taxation term considering the reference user. Thus, the complexity of power allocation for REFIM is basically the same as that for WF. MGR in \cite{refStolyar2}, one of the state-of-the-art dynamic IM algorithms, adjusts the power allocation for every $n_p>1$ slots (relatively slowly) and condenses the complexity for updating power to $1/n_p$. However, MGR has additional complexity from a virtual scheduling that needs to be run $n_v$ times per slot. Accordingly, the total computational complexity for power allocation is high, $\frac{1}{n_p} \mathcal{O}(S) + n_v \mathcal{O}(SK)$.
Another recently developed MC-IIWF in \cite{refVenturinoCoordinated} is the centralized algorithm that has iteration loops for power allocation and user scheduling. Let $T_1$ be the number of iterations needed for iteration loops. Then, the total computational complexity is equal to $T_1\cdot\left(\mathcal{O}(SK)+\mathcal{O}(SN)\right)$.

Now let us investigate the inter-BSs signaling complexity. EQ and WF do not require any inter-BS message passing overhead because they are autonomous algorithms  without considering neighboring BSs, but at the cost of performance reduction, as will be shown in Section~\ref{section_simulation}. MGR adjusts the power allocation slowly so that it requires not per-slot but periodic feedback, $|N_{n}| S$ (sensitivity information for neighboring BSs and all subchannels). MC-IIWF, a centralized algorithm, assumes a central control unit to have complete information. REFIM requires the periodic feedback about the candidate users for the reference users, $\rho |\mathcal{K}_n| A S$, where $\rho$ is the average percentage of edge users and $A=4$ is the number of required information {(F0)$\sim$(F3)} about the reference user. Note that while macro BSs require the per-slot feedback for the indexes of scheduled users at each subchannel, femto BSs do not.

In summary, the computational complexity of REFIM is the same as that of WF and is much lower than that of state-of-the-art dynamic IM algorithms such as MGR and MC-IIWF. For signaling complexity, although the feedback per slot-wise manner is challenging, the information needs to be exchanged only between neighboring macro BSs at each slot are just the indexes of scheduled users. We believe that such small amount of information can be easily exchanged through high-speed dedicated backbones.

%%% Local Variables:
%%% mode: latex
%%% TeX-master: "JSAC"
%%% End:

%===================================================================
\section{Performance Evaluation}\label{section_simulation}
%===================================================================
We verify the system performance through extensive simulations under various topologies and scenarios. First, in order to verify the effectiveness of REFIM by varying several tunable parameters and comparing it with other algorithms, a two-tier macro-cell networks composed of 19 hexagonal cells is considered in subsection \ref{subsection_simulation1}. Second, in order to provide more realistic simulation results, a real 3G BS deployment topology consisting of heterogeneous environments (urban, suburban and rural areas) is considered in subsection \ref{subsection_simulation2}. Third, a heterogeneous network topology with small cells inside macro cells is also considered in subsection \ref{subsection_simulation3}.

In our simulations, macro and small cells are loaded with 20 and 4 users, respectively and they are uniformly distributed in each cell. All users are assumed to have a logarithmic utility function, i.e., $U(R_k) = \log{R_k}$, but the other utility functions enforcing more fairness are also considered in our technical report \cite{refSonTechReport}. We consider a system having 16 subchannels each of which consists of multiple subcarriers. The maximum transmit powers of macro and small BSs are 43dBm and 15dBm \cite{refWINNER}, respectively. In modeling the propagation environment, an ITU PED-B path loss model $16.62+37.6\log_{10}(d[m])$ for macro cells and an indoor path loss model $37+32\log_{10}(d[m])$ for small cells with 10dB penetration loss due to walls are adopted. Jakes' Rayleigh model (with the speed of 3km/h), where the channel coefficients vary continuously over slots, is adopted for fast fading. The channel bandwidth and the time-slot length are set to be 10MHz and 1ms, respectively.
%Jakes' Rayleigh model is also adopted for modeling the fast fading.
%The other parameters for simulations follow the suggestions in the IEEE 802.16m evaluation methodology document \cite{ref802.16m_EMD}.

REFIM is compared to conventional EQ and WF, as well as MGR \cite{refStolyar2} and MC-IIWF \cite{refVenturinoCoordinated} developed recently. As performance metrics, the geometric average of user throughputs (GAT) and the average of edge user throughputs (AET) are used. We use GAT since maximizing this metric is equivalent to our system objective (sum of log throughputs). We consider AET as the average of the bottom 5\% of user throughputs (i.e., 5th percentile throughput), which can be regarded as a representative performance metric of cell-edge users.
%We consider AET as the average of the lowest 5\% throughput of users, which can be regarded as a representative performance metric of cell-edge users.

\subsection{Effectiveness of Proposed Algorithm}\label{subsection_simulation1}
Fig. \ref{fig:verification_a} shows the GAT performance of REFIM for the number of reference users per subchannel and that of EQ as a baseline. As mentioned in Remark \ref{remark:taxation}, REFIM without a reference user is reduced as WF. REFIM taking reference users into consideration can obtain higher performance gain. It is noteworthy that considering only the one reference user per subchannel is efficient enough because it can obtain more than 97\% of the performance considering all the six neighbors. Fig. \ref{fig:verification_b} shows the effect of iteration loops and initial power: (i) adding user scheduling loop and/or power allocation loop gives additional performance gain from any initial power setting and (ii) using power at the previous slot as an initial power outperforms other two strategies. However, for the case in which power at the previous slot is used as an initial power, the performance gain from adding power allocation loop is marginal. Fig. \ref{fig:verification_b} is an encouraging result that leads us to design the algorithm without loops and use the previous power as an initial power.

\begin{figure}[t!] \centering
\subfigure[The number of reference users]
{\hspace{-0.1cm}\epsfig{file=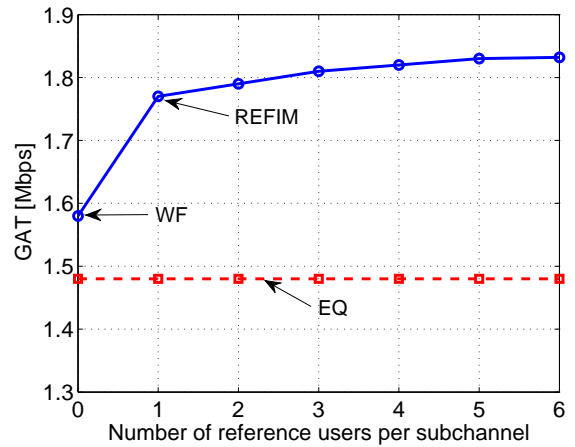,width=0.4\textwidth}\label{fig:verification_a}}\\
\subfigure[Iteration loops and initial power]
{\hspace{-0.1cm}\epsfig{file=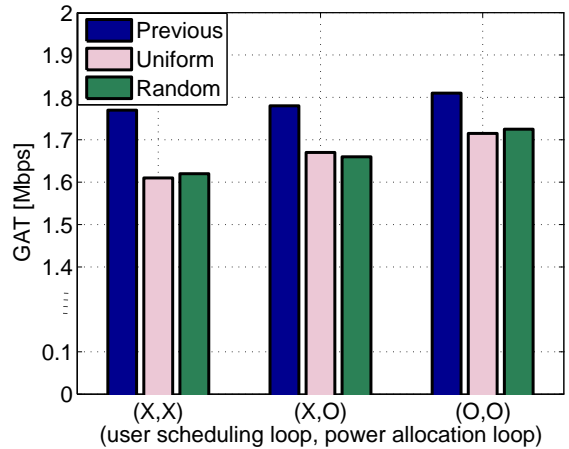,width=0.4\textwidth}\label{fig:verification_b}}
\caption{Effect of several tunable parameters.}
\label{fig:verification}
\end{figure}

\begin{figure}[t!]\centering
\subfigure[Time series of transmit powers for different subchannels]
{\epsfig{file=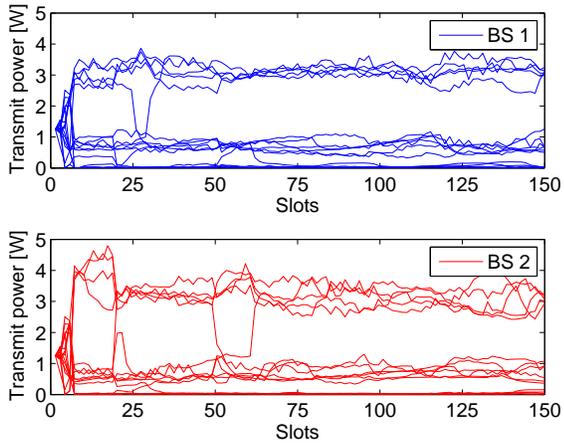,width=0.4\textwidth}\label{fig:power_level_a}}\\
\subfigure[Average transmit powers between 75 and 125 slots]
{\hspace{-0.6cm}
\epsfig{file=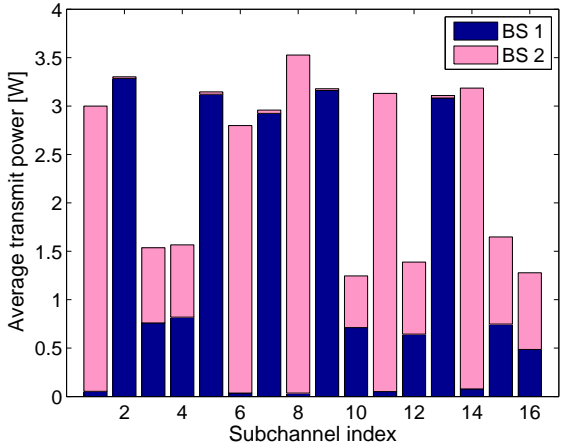,width=0.4\textwidth}\label{fig:power_level_b}}\\
\caption{Transmit power levels for different subchannels in a linear two-cell network.}
\label{fig:power_level}
\end{figure}

We conjecture that this is because in some sense using the previous power rule has the effect of iterations not in a slot but over a series of multiple slots. To verify this statement, we provide additional simulations in a linear two-cell network where the distance between BSs is 2km. Each cell has two groups of users: center and edge users, whose distances from the associated BS are within 200-400m and 700-900m, respectively. We adopt the simple network configuration to gain insight easily as well as for easy of presentation, but all discussions can be extended to the general cases.

Fig. \ref{fig:power_level_a} shows the time-series of transmit powers on different subchannels. Although the powers do not seem to quite converge, they remain the certain levels for several dozens of slots. In Fig. \ref{fig:power_level_b}, we plot the average transmit power levels on different subchannels during the period between 75 and 125 slots. As can be clearly seen, each BS exclusively utilizes five subchannels of high powers and shares six subchannels of low powers with the other. We further investigate the relationship between the user groups and their subchannels on which they are scheduled. Interestingly, for the most of time (more than 98\% of slots), the center and edge users are served by the set of subchannels with low and high powers, respectively. In other words, it is highly likely that on each subchannel a user will be selected by the scheduler who has similar channel conditions to the user selected at the previous slot. Thus, given similar user scheduling over consecutive slots, the power allocation with the previous power can be interpreted as if it has iterations over a series of multiple slots. This explains why REFIM that executes user scheduling and power allocation step-by-step in a slot can achieve a good solution without much performance degradation.

\smallskip
In Fig. \ref{fig:feedback_period}, we also test the effect of outdated feedback information about reference users. We consider two types of users with different speeds: nomadic users (or stationary users that have fixed path loss and shadowing factor, but have time-varying fast-fading) and mobile users (moving fast with speed of 60km/h). For nomadic users, the GAT performance degradation is relatively small, even though we choose a long feedback period such as 200 slots. For mobile users, the GAT performance naturally decreases due to the error of feedback information as the feedback period increases. In practical systems, different types of users with different speeds are expected to coexist. In such an environment, to reduce the amount of feedbacks while maintaining the performance degradation marginal, it is essential to adaptively control the period of feedback for the different speed users, e.g., $T=200$ slots for nomadic users ($\sim$3km/h) and $T=10$ slots for mobile users ($\sim$60km/h)\footnote{We use the simplified version of random waypoint model \cite{refBrochRWP}, where each user starts from an initial point, randomly chooses its destination, and moves toward it at a given speed. After reaching the destination, it repeats this process unite the end of simulation time.}.

\begin{figure}[t]
\centering \epsfig{file=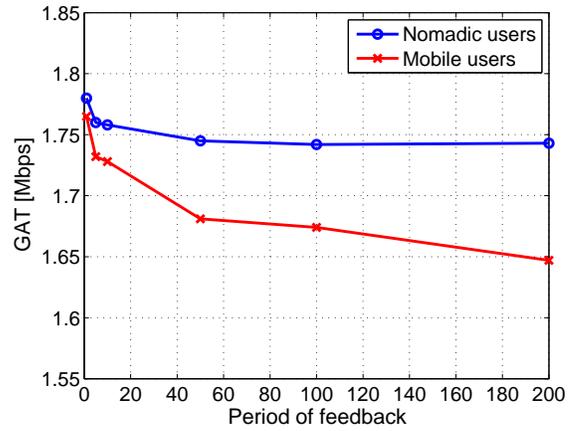,width=0.4\textwidth}
\caption{The effect of outdated feedback information.}
\label{fig:feedback_period}
\end{figure}

\begin{figure}[t]
\centering
\epsfig{file=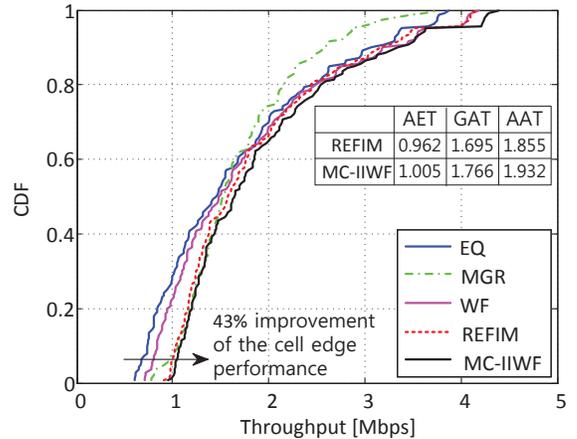,width=0.4\textwidth}
\caption{Comparison with other algorithms.}
\label{fig:comparison_symmetric}
\end{figure}

\subsection{Performance Comparison with Other Algorithms}
Now we compare the performance of REFIM with other four algorithms: conventional EQ, selfish WF, MGR \cite{refStolyar2} which adjusts power allocation infrequently, compared to per-slot basis user scheduling, and MC-IIWF \cite{refVenturinoCoordinated} which is a centralized algorithm achieving the near-optimal performance. Fig. \ref{fig:comparison_symmetric} shows the CDF (cumulative distribution function) of the throughput of entire users in the network for different algorithms. Compared to EQ, WF and MGR, REFIM can improve the throughputs for all users in the network. Particularly, we can observe higher improvement (43\% improvement in AET compared to EQ) for users achieving low throughputs, i.e., users at cell edges. This is due to the fact that IM is mainly targeted for performance improvement of cell-edge users. In addition, REFIM can achieve about 95\% of the performance of near-optimal MC-IIWF in terms of two representative throughput metrics (GAT and AET) as well as the arithmetic average of user throughputs (AAT). It is somewhat surprising that such a simple distributed algorithm can obtain a similar performance to the centralized algorithm that is hard to implement due to prohibitive complexity.

\subsection{Topology with Real BS Deployment: Urban, Suburban and Rural Environments}\label{subsection_simulation2}

\begin{figure}[t] \centering
\vspace{0.3cm}\hspace{-0.15cm}
\epsfig{file=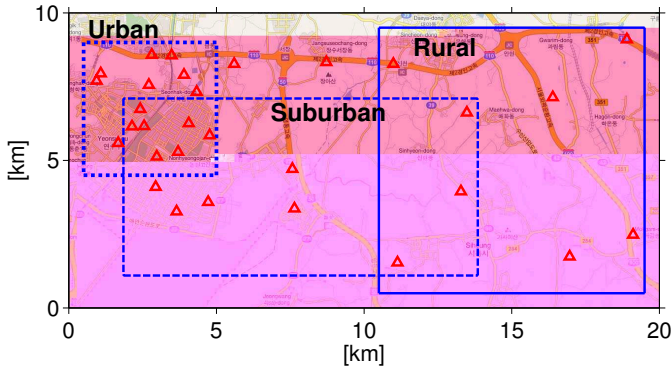,width=0.48\textwidth}
\caption{Real 3G BS deployment map.}
\label{fig:real_topology1}
\end{figure}

Fig. \ref{fig:real_topology1} depicts the map of BS layout that we use for more realistic simulations. It is a part of real 3G network operated by one of the major mobile network operators in Korea. There are a total of 30 BSs within $20\times10\;\textrm{km}^2$ rectangular area. We assume that the number of BSs per unit area is proportional to the user density. In other words, the average number of users per cell is almost similar because BSs in an urban environment cover a small area and BSs in a rural environment a large area. Under this assumption, we generate users one-by-one in the rectangular area and attach them to the closest BS until each BS will have 20 users. We choose this partial map to include a challenging scenario that three environments are mixed together. We tested several other maps, and obtained similar or even better performance of REFIM.

\begin{figure}[t] \centering
\subfigure[GAT (geometric average of user throughputs)]
{\epsfig{file=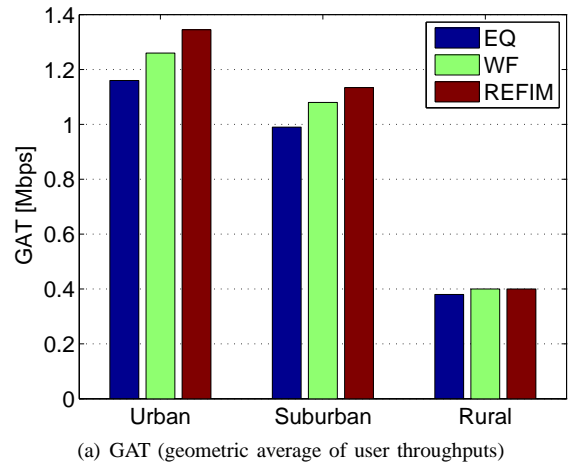,width=0.4\textwidth}\label{fig:real_topology_GAT}}\\
\subfigure[AET (average of edge user throughputs)]
{\epsfig{file=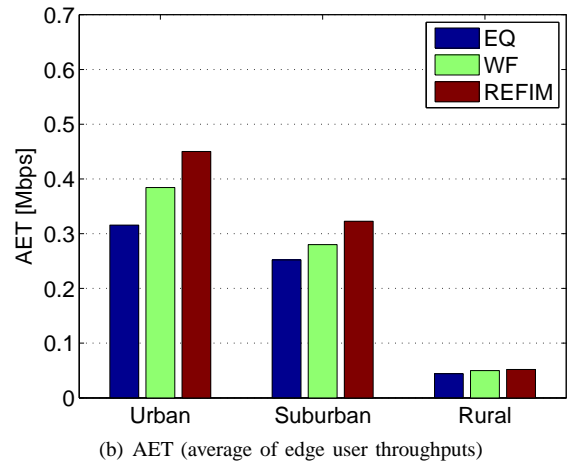,width=0.4\textwidth}\label{fig:real_topology_AET}}%
\caption{Throughput performances in real BS deployment topology.}
\label{fig:real_topology_throughput}
\end{figure}

\begin{figure}[h!] \centering
\epsfig{file=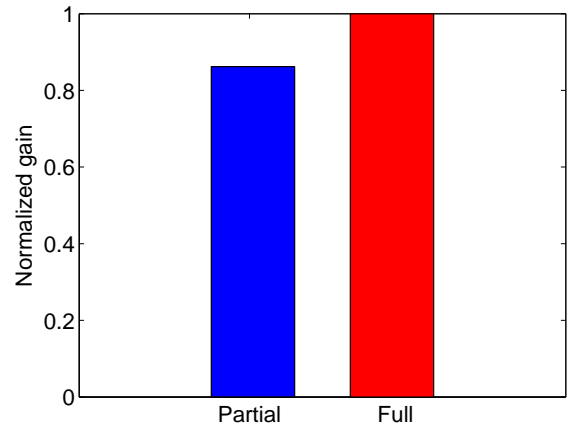,width=0.4\textwidth}
\caption{Effect of partial deployment.}
\label{fig:partial_deployment}
\end{figure}

We examine three different zones: urban (15 BSs in $4.5\times4.5\;\textrm{km}^2$), suburban (15 BSs in $12\times6\;\textrm{km}^2$) and rural (8 BSs in $9\times9\;\textrm{km}^2$) areas\footnote{In order to see clearly how the density of BSs affects the performance gain, the distance between BSs in suburban and rural zones are increased by 1.5 and 2 times, respectively.}. Fig. \ref{fig:real_topology_throughput} shows GAT and AET performance under urban, suburban and rural environments. As expected, we can obtain high performance improvements in the urban and suburban environments. However, almost low or no gain is found in the rural environment, which means that the IM does not take much effect in a sparse topology, which follows our intuition. For example, in the urban area, the performance gains of REFIM are 16\% and 42\% in terms of GAT and AET, respectively.

\begin{figure*}[t]
\centering
\subfigure[Heterogeneous network topology]
{\epsfig{file=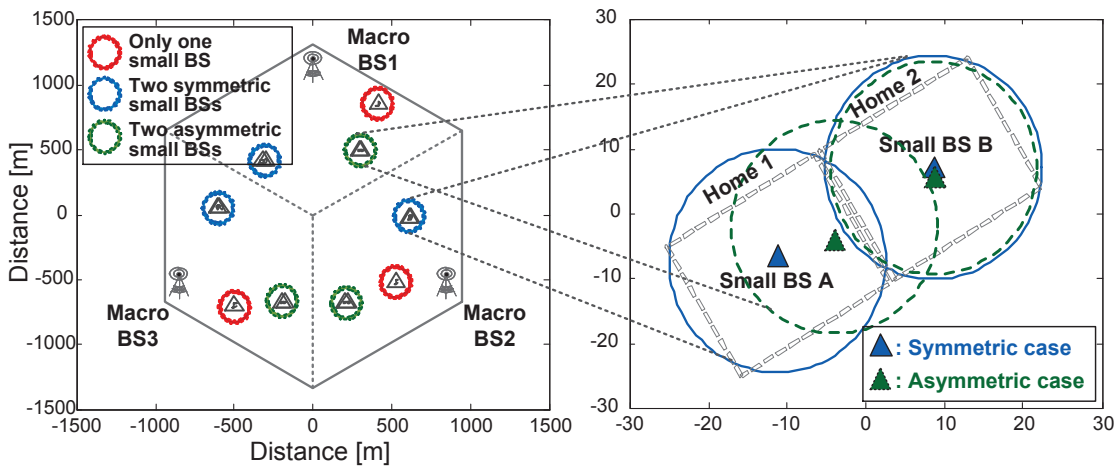,width=0.8\textwidth}\label{fig:heterogeneous_topology}}\\%
\subfigure[Five small cells per a macro cell]
{\hspace{0.5cm}\epsfig{file=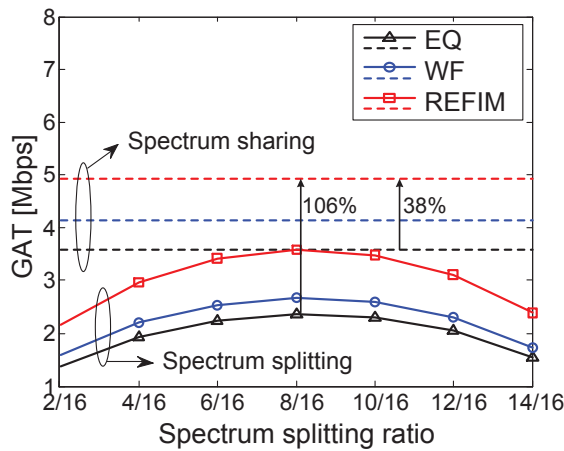,width=0.4\textwidth}\label{fig:femto5}}%
\subfigure[Ten small cells per a macro cell]
{\epsfig{file=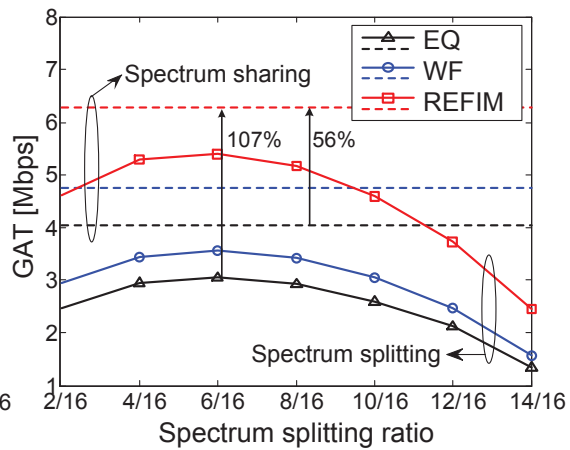,width=0.4\textwidth}\label{fig:femto10}}%
\caption{Performance in the heterogeneous network topology.}
\label{fig:femto}
\end{figure*}

Another nice feature of REFIM is \textit{incremental deployment}. Suppose that we implement our algorithm only on the BSs in a specific area. While the BSs inside this area performs well as we want, the BSs in the boundary of the area does not. This is because they may not receive information about reference users from the some of its neighboring BSs on which our algorithm is not implemented. Even in such a case, our algorithm will automatically reduce to WF. Thus, it performs like WF at least and better than EQ. Compared to the full deployment case, the partial deployment case where only 15 BSs (mainly selected from the urban areas among 30 BSs) are equipped with REFIM can achieve more than 85\% gain as shown in Fig. \ref{fig:partial_deployment}. The result encourages the mobile network operator to upgrade its BSs incrementally from the urgent ones, e.g., densely located BSs experiencing heavy interference, and thus low capacity.

\subsection{Heterogeneous Networks (Macro + Small BSs)}\label{subsection_simulation3}
Now we consider a heterogeneous network topology having several small cells inside macro cells as shown in Fig. \ref{fig:heterogeneous_topology}. As a kind of small BSs, we consider the femto BSs that are deployed and provide a high-speed indoor access mainly to home users. We need to reflect femto-to-femto cell interference as well as macro-to-macro and macro-to-femto cell interferences. To this end, we consider the mixture of three femto BSs deployment cases: (i) only one femto BS case (no femto-to-femto interference), (ii) two symmetric femto BSs case (strong femto-to-femto interference): two femto cells are adjacent with each other and each femto BS is located in the center of the home, (iii) two asymmetric femto BSs case (very strong femto-to-femto interference): two femto cells are adjacent with each other and femto BS 1 in home 1 is located at the border between homes. The last case can often happen because users locate their own femto BSs wherever they want without considering next door neighbors.

In Figs. \ref{fig:femto5} and \ref{fig:femto10}, we compare the performance of the spectrum sharing policy (i.e., universal frequency reuse) between macro and femto cells with that of the spectrum splitting policy where macro and femto cells orthogonally use the resource. Note that the performance curves for the spectrum sharing policy and the spectrum splitting policy are represented by solid and dotted lines, respectively. The x-axis represents the ratio of subchannels used by macro cells among all 16 subchannels.

In the case of five femto cells per a macro cell in Fig. \ref{fig:femto5}, there exists small cross-tier (macro-to-femto) interference. Thus, the spectrum sharing policy is always better than the spectrum splitting policy. For example, even the performance of EQ without any interference management in the spectrum sharing policy is higher than or equal to the performance of REFIM in the optimal spectrum splitting (at 8/16). If the number of femto cells increases, then the portion of macro users who will see more and closer femto cells increases. Consequently, it is highly probable that their performances are degraded by the severe cross-tier interference. As shown in Fig. \ref{fig:femto10} with ten femto cells per a macro cell, the best performance in the optimal spectrum splitting (at 6/16) can catch up with the that of EQ and WF in the spectrum sharing policy. However, if the proposed IM algorithm, e.g., REFIM, is adopted in the spectrum sharing policy, then we can mitigate the cross-tier interference, resulting in the better performance than any case in the spectrum splitting policy. Note that the performance improvements of our REFIM in the spectrum sharing policy compared to EQ in the spectrum sharing and spectrum splitting policies are 38\% and 106\% in Fig. \ref{fig:femto5}, and they become larger as the number of femto cells increases, i.e., 56\% and 107\% in Fig. \ref{fig:femto10}.

The spectrum splitting policy is expected to be widely used rather than the spectrum sharing policy in an early stage of femto cell deployments because it makes the femto cells easily coexist with macro cells without worrying about the macro-to femto interference. However, the results in Fig. \ref{fig:femto} enlighten us on the potential gain of the spectrum sharing policy. Therefore, we believe that, if the IM algorithms become more mature in the near future, then the spectrum sharing policy will be adopted in order to maximally exploit the spectral resources due to the explosive traffic demands

%Although the spectrum splitting policy is expected to be widely used in an early stage of femto cell deployments because it makes the femto cells easily coexist with macro cells. However, the results in Fig. \ref{fig:femto} enlighten us on the potential gain of the spectrum sharing policy.

%===================================================================
\section{Conclusion}\label{section_conclusion}
%===================================================================
Heterogeneous access networks, consisting of cells with different sizes and ranging from macro to femto cells, will play a pivotal role in the next-generation broadband wireless network. They can increase the network capacity significantly to meet the explosive traffic demand of users with limited capital/operating expenditures and spectrum constraints. One of the biggest challenges in such environments is how to effectively manage interferences between heterogeneous cells. To tackle this challenge, this paper developed REFIM, which is an efficient low-complex and fully distributed IM in downlink heterogeneous multi-cell networks. Our key idea is to use the notion of reference user, which can spatially simplify the impact of all other neighboring cells by a single virtual user and result in the power control algorithm with low computational and signaling overhead. In order for REFIM to be implemented even on the femto BSs, we also further reduced the feedback over backhauls both temporally and spatially. Through extensive simulations and complexity analysis, we demonstrated that REFIM not only performs well but also is practically implementable. We also concluded that as long as appropriate IM algorithms such as REFIM are adopted, the spectrum sharing policy can outperform the best spectrum splitting policy where the number of subchannels is optimally divided between macro and femto cells.

\section*{Acknowledgments}
The authors would like to thank the anonymous reviewers for their helpful comments that greatly improved the quality of this paper. The authors also would like to thank Prof. Mung Chiang, Prof. Jianwei Huang and Dr. Paschalis Tsiaflakis for their helpful discussions.

\bibliographystyle{IEEEtran}
\bibliography{IEEEabrv,bib}

\vspace{-0.5cm}
\begin{biography}[{\includegraphics[width=1in,height=1.25in,clip,keepaspectratio]{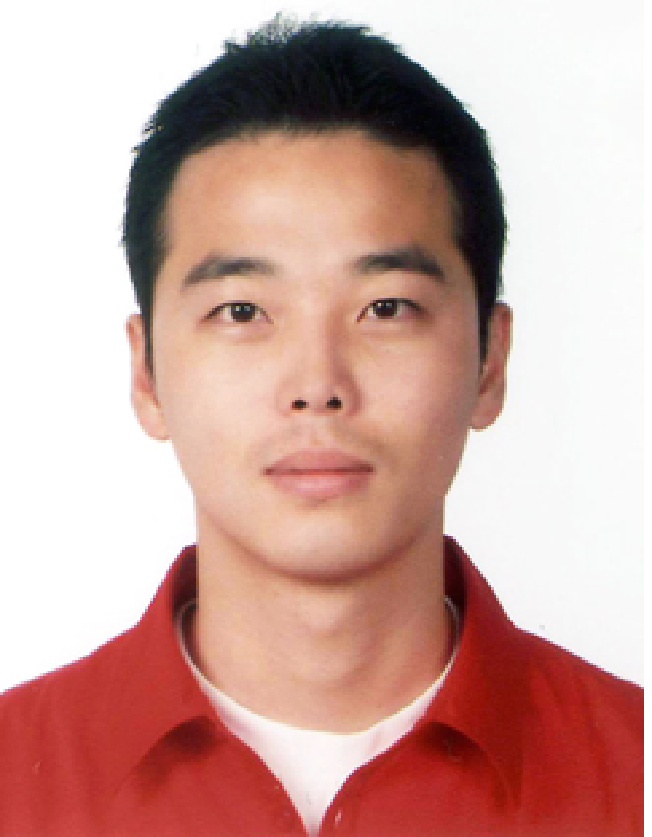}}]{Kyuho Son} (S'03-M'10) received his B.S., M.S. and Ph.D. degrees all in the Department of Electrical Engineering from Korea Advanced Institute of Science and Technology (KAIST), Daejeon, Korea, in 2002, 2004 and 2010, respectively. He is currently a post-doctoral research associate in the Department of Electrical Engineering at the University of Southern California, CA. His current research interests include interference management in heterogeneous cellular networks, green networking and network economics. He was the Web Chair of the 7th International Symposium on Modeling and Optimization in Mobile, Ad Hoc, and Wireless Networks (WiOpt 2009).
\end{biography}

\vspace{-1cm}
\begin{biography}[{\includegraphics[width=1in,height=1.25in,clip,keepaspectratio]{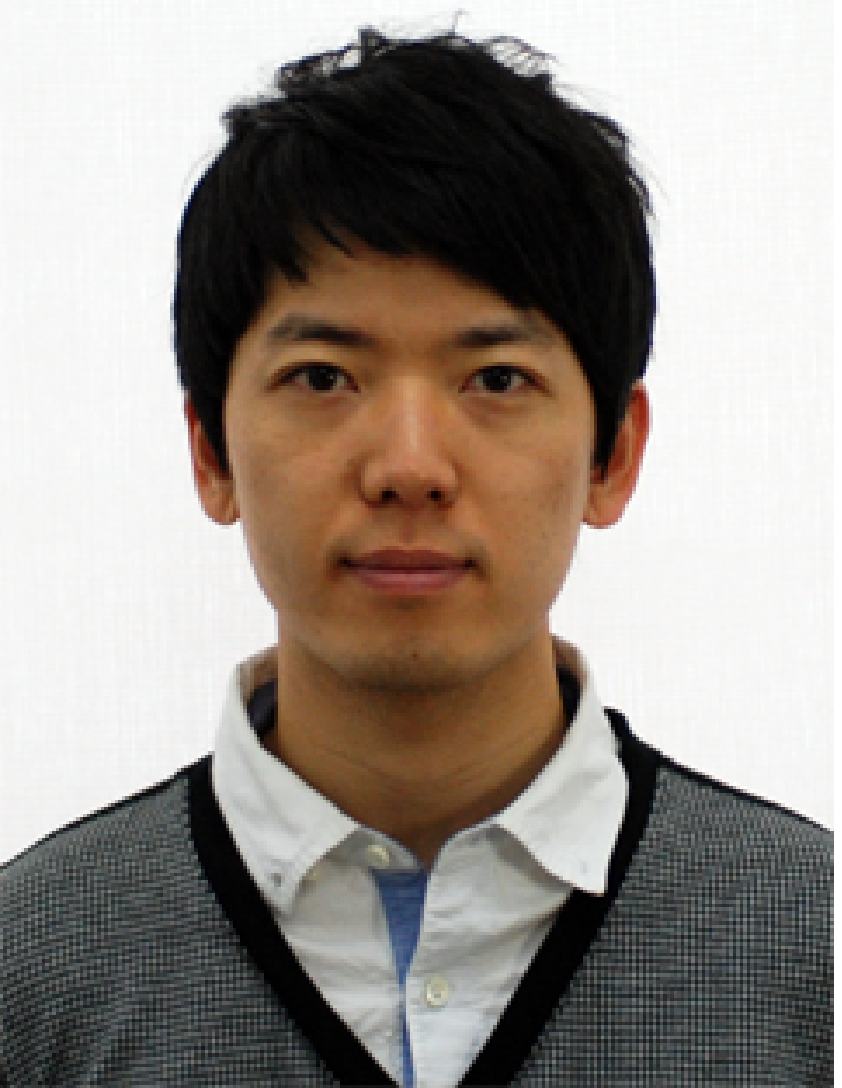}}]{Soohwan Lee} (S'11)
received his B.S. degree in the School of Electrical Engineering and Computer Science from Kyungpook National University, South Korea, in 2009, and his M.S. degree in the Department of Electrical Engineering from Korea Advanced Institute of Science and Technology (KAIST), South Korea, in 2011. He is currently a Ph.D student in the Department of Electrical Engineering at KAIST. His current research interests include interference management in heterogeneous cellular networks, green wireless networking, and security management in cellular networks.
\end{biography}

\vspace{-1cm}
\begin{biography}[{\includegraphics[width=1in,height=1.25in,clip,keepaspectratio]{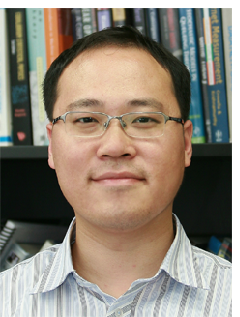}}]{Yung Yi} (S'04-M'06)
received his B.S. and the M.S. in the School of Computer Science and Engineering from Seoul National University, South Korea in 1997 and 1999, respectively, and his Ph.D. in the Department of Electrical and Computer Engineering at the University of Texas at Austin in 2006. From 2006 to 2008, he was a post-doctoral research associate in the Department of Electrical Engineering at Princeton University. Now, he is an assistant professor at the Department of Electrical Engineering at KAIST, South Korea. His current research interests include the design and analysis of computer networking and wireless communication systems, especially congestion control, scheduling, and interference management, with applications in wireless ad hoc networks, broadband access networks, economic aspects of communication networks, and green networking systems. He has been serving as a TPC member at various conferences such as ACM Mobihoc, Wicon, WiOpt, IEEE Infocom, ICC, Globecom, ACM CFI, ITC, the local arrangement chair of WiOpt 2009 and CFI 2010, and the networking area track chair of TENCON 2010.
\end{biography}

\begin{biography}[{\includegraphics[width=1in,height=1.25in,clip,keepaspectratio]{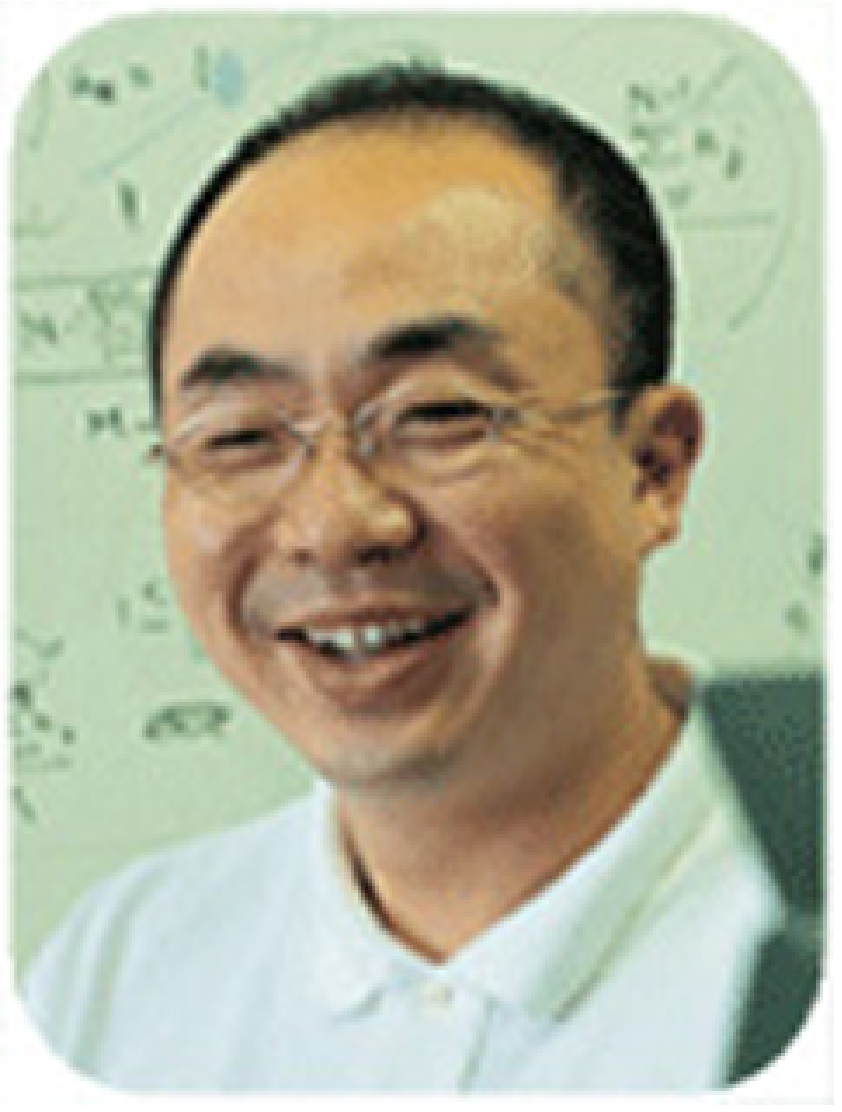}}]{Song Chong} (S'93-M'95) received the B.S. and M.S. degrees in Control and Instrumentation Engineering from Seoul National University, Seoul, Korea, in 1988 and 1990, respectively, and the Ph.D. degree in Electrical and Computer Engineering from the University of Texas at Austin in 1995. Since March 2000, he has been with the Department of Electrical Engineering, Korea Advanced Institute of Science and Technology (KAIST), Daejeon, Korea, where he is a Professor and the Head of the Communications and Computing Group of the department. Prior to joining KAIST, he was with the Performance Analysis Department, AT\&T Bell Laboratories, New Jersey, as a Member of Technical Staff. His current research interests include wireless networks, future Internet, and human mobility characterization and its applications to mobile networking. He has published more than 100 papers in international journals and conferences.

He is an Editor of Computer Communications journal and Journal of Communications and Networks. He has served on the Technical Program Committee of a number of leading international conferences including IEEE INFOCOM and ACM CoNEXT. He serves on the Steering Committee of WiOpt and was the General Chair of WiOpt '09.  He is currently the Chair of Wireless Working Group of the Future Internet Forum of Korea and the Vice President of the Information and Communication Society of Korea.
\end{biography}

\end{document}